\documentclass[10pt]{amsart}
\usepackage{mathpple}
\usepackage{amsmath}                 
\usepackage{amssymb}                 
\usepackage{amsbsy}
\usepackage{amsthm}
\usepackage{amsfonts}
\usepackage{mathrsfs}                
\usepackage{dsfont}
\usepackage{bm}                      
\usepackage{bbding,manfnt}           
\usepackage{mathtools,leftidx}
\usepackage{stmaryrd}

\usepackage{graphicx}                
\usepackage{color,xcolor}            
\usepackage{booktabs}                
\usepackage{eso-pic}                 
\usepackage{tikz}
\usetikzlibrary{arrows,backgrounds,shadows,decorations.markings,matrix}
\usepackage{hyperref}
\usepackage{marginnote}
\usepackage{titlesec}
\titleformat{\section}{\large\bfseries}{\thesection}{1em}{}[\vspace*{8pt}\hrule]
\titleformat{\subsection}{\bfseries}{\thesubsection}{1em}{}

\usepackage{setspace}
\renewcommand\appendix{\setcounter{secnumdepth}{-2}}

\usepackage[toc,page]{appendix}

\usepackage{titletoc}

\titlecontents{section}[0mm]                                  
{\fontsize{8pt}{20pt}}{\contentspush{\thecontentslabel\ }}    
{}{\titlerule*{.}\contentspage}
\titlecontents{subsection}[4mm]
{\fontsize{6pt}{20pt}}{\contentspush{\thecontentslabel\ }}
{}{\titlerule*{.}\contentspage}

\xdefinecolor{tsinghua}{rgb}{0.455,0.204,0.506}
\newcommand{\watermark}[3]{\AddToShipoutPictureBG{
\parbox[b][\paperheight]{\paperwidth}{
\vfill
\centering
\tikz[remember picture, overlay]
  \node [rotate = #1, scale = #2] at (current page.center)
    {\textcolor{tsinghua!20}{#3}};
\vfill}}}

\usepackage{fancyhdr}
\makeatletter

\newcommand{\f}[2]{\frac{#1}{#2}}

\newcommand{\Rmnum}[1]{\mathrm{\expandafter\@slowromancap\romannumeral #1@}}
\newcommand{\df}{\mathsf{d}}
\newcommand{\id}{\mathsf{id}}

\newcommand{\ha}{\alpha}
\newcommand{\hb}{\beta}
\newcommand{\p}{\partial}
\newcommand{\pd}[2]{\frac{\partial#1}{\partial#2}}

\renewcommand{\hom}{\mathrm{Hom}}

\newcommand{\wt}{\mathsf{wt}}

\newcommand{\ainfty}{$\mathsf{A_\infty}$-}

\newcommand{\linfty}{$\mathsf{L_\infty}$-}

\newcommand{\e}{\mathsf{e}}

\newcommand{\ho}{\mathsf{H}}

\newcommand{\lra}[1]{\langle #1 \rangle}

\newcommand{\lr}[1]{\left(#1\right)}
\providecommand{\abs}[1]{\left\lvert#1\right\rvert}

\makeatother
\theoremstyle{plain}
\newtheorem*{mthm}{Theorem}
\newtheorem{thm}{Theorem}[section]

\newtheorem{lem}[thm]{Lemma}
\newtheorem{coro}[thm]{Corollary}
\theoremstyle{definition}

\theoremstyle{remark}
\newtheorem*{rem}{Remark}

\hfuzz=\maxdimen               
\newcommand{\hh}{\mathrm{HH}}

\newcommand{\hp}{\mathrm{HP}}

\newcommand{\fix}{\mathrm{Fix}}
\newcommand{\jac}{\mathrm{Jac}}
\newcommand{\ks}{\mathsf{KS}}

\begin{document}
\numberwithin{equation}{section}
\tikzset{->-/.style={decoration={
  markings,
  mark=at position #1 with {\arrow{>}}},postaction={decorate}}}
\tikzset{->>-/.style={decoration={
  markings,
  mark=at position #1 with {\arrow{>>}}},postaction={decorate}}}
\watermark{60}{10}{}
\title{\bfseries{Unfolding of Orbifold LG B-models: a case study}}
\author{Weiqiang He, Si Li, Yifan Li}
\date{}
\newcommand{\Addresses}{{
  \bigskip
  \footnotesize

  Weiqiang He, \textsc{Department of Mathematics, Sun Yat-sen University
    Guangzhou, 510275, China}\par\nopagebreak
  \textit{E-mail address}: \texttt{hewq@mail2.sysu.edu.cn}

  \medskip

  Si Li, \textsc{Department of Mathematical Sciences and Yau Mathematical Sciences Center, Tsinghua University,
    Beijing 100084, China}\par\nopagebreak
  \textit{E-mail address}: \texttt{sili@mail.tsinghua.edu.cn}

  \medskip

  Yifan Li, \textsc{Yau Mathematical Sciences Center, Tsinghua University,
    Beijing 100084, China}\par\nopagebreak
  \textit{E-mail address}: \texttt{yf-li14@mails.tsinghua.edu.cn}

}}
\maketitle
\begin{abstract} In this note we explore the variation of Hodge structures associated to the orbifold Landau-Ginzburg B-model whose superpotential has two variables. We extend the  Getzler-Gauss-Manin connection to Hochschild chains  twisted by group action. As an application, we provide explicit computations for the Getzler-Gauss-Manin connection on the universal (noncommutative) unfolding of $\mathds{Z}_2$-orbifolding of A-type singularities. The result verifies an example of deformed version of Mckay correspondence.

\end{abstract}

\thispagestyle{plain}
\pagestyle{fancy}

\section{Introduction}
Associated to a triple $(A,W,G)$, where $A$ is an associative algebra over $\mathds{C}$ with a compatible $G$-action and $W$ is a $G$-invariant central element of $A$, we consider a curved algebra $A_W[G] \coloneqq A \rtimes \mathds{C}[G]$ with $W$ as a curvature. In this note, we investigate the deformation theory and Hodge structures for a certain type of such curved algebras.

In \cite{HLL18}, we have shown that the compact type Hochschild cohomology $\hh^\bullet_c(A_W,A_W[G])^G$ is isomorphic to $\hh^\bullet_c(A_W[G],A_W[G])$ as Gerstenhaber algebras. As a consequence, the deformation of $A_W[G]$ is controlled by the differential graded Lie algebra (\textsf{dgLa}) of the Hochschild cochains $(C^\bullet_c(A_W,A_W[G])^G,\delta_b,\{\,,\,\}$. In this paper, we study  polynomial algebras in two variables $A = \mathds{C}[x,y]$. This includes orbifold ADE singularities as our main interest in this paper. Obstruction theoretical computation shows that the relevant \textsf{dgLa} is un-obstructed, leading to a smooth formal moduli space $\mathcal M$ which is locally parameterized by the Hochschild cohomology $HH^\bullet(A_W[G],A_W[G])$.

Our study of this moduli space $\mathcal M$ is motivated by Saito's work \cite{S83} on isolated singularities, which is related to so-called Landau-Ginzburg (LG) B-models in modern terminology. In \cite{S83}, it was shown that the deformation space of an isolated singularity carries a version of variation of polarized Hodge structures with semi-infinite filtrations. It leads to an integrable structure on the tangent bundle of the moduli, which is nowadays called Frobenius manifold. This Frobenius manifold structure plays a central role in topological field theories, especially in Gromov-Witten type theories. For example, the data of Frobenius manifold on the deformation of isolated polynomial singularities is mirror to the data of counting solutions of Witten's equation on Riemann surfaces, a theory known as Fan-Jarvis-Ruan-Witten (FJRW) theory  \cite{FJR13} for Landau-Ginzburg A-models.  However, Saito's construction only involves `un-orbifold' cases $(A,W,G= \lra{1})$, while the full mirror symmetry between Landau-Ginzburg models asks for all orbifold groups. This requires the construction and computation of Frobenius manifold structure on the aforementioned moduli space $\mathcal M$.

 Barannikov \cite{B01,B02} and Barannikov-Kontsevich \cite{BK} introduced the important
notion of (polarized) variation of semi-infinite Hodge structures (VSHS), generalizing Saito's framework to many other geometric contexts and non-commutative world \cite{B01,KKP08}. Following this route, we shall consider the period cyclic homology of a deformed algebra of $A_W[G]$, with a Hodge filtration induced by the cyclic parameter $u$ and the flat Gauss-Manin connection constructed by Getzler \cite{G93}. They give rise to a flat bundle over the moduli space $\mathcal M$, carrying important data of Hodge filtration. In this note, we establish a version of the Getzler-Gauss-Manin connection via operations of $G$-twisted cochains $C^\bullet_c(A,A[G])^G$ acting on the $G$-twisted chains $C^c_\bullet(A,A[G])_G$.  This encodes the same information as the Getzler-Gauss-Manin connection on the deformation space of the algebra $A[G]$, but is easier to compute in practice. As an application, we perform a case study for orbifold A-type singularity $(A_{2n-1},\mathds{Z}_2)$. We find (see Theorem \ref{main-thm}),

\begin{mthm}
Consider an orbifold LG B-model $(A,W,G)$ with $A=\mathds{C}[x,y]$, $W$ invertible and finite $G\subset \mathrm{SL}(2,\mathds{C})$ acting diagonally on $\mathds{C}^2$. The moduli space $\mathcal{M}$ of miniversal deformations of $A_W[G]$ is smooth, equipped with a variation of semi-infinite Hodge structures (VSHS) given by a flat vector bundle of period cyclic homologies. In this fashion, there is an isomorphism between the moduli spaces associated to $(\mathds{C}[x,y],x^{2n}+y^2,\mathds{Z}_2)$ and $(\mathds{C}[z,w],{z^n+zw^2},\lra{1})$, which is compatible with the VSHS's on them.
\end{mthm}

It can be seen as an example of Mckay correspondence for LG models \cite{Q09}, but involves the deformation data. Here, $( \mathds{C}[x,y],{x^{2n}+y^2},\mathds{Z}_2)$ is associated to the $A_{2n-1}$-singularity $W = x^{2n}+y^2$ on an orbifold $X = \mathds{C}^2/\mathds{Z}_2$ and $(\mathds{C}[z,w],{z^n+zw^2},\lra{1})$ is associated to the $D_{n+1}$-singularity $\tilde{W} = z^n+zw^2$ on $\mathds{C}^2 \overset{i}{\hookrightarrow} Y$, where $\pi\colon Y \to X$ is the minimal resolution (so it is crepant) and $\tilde{W} = i^*\circ \pi^*(W)$.

There are three directions of generalizations of such a correspondence.  One is for more general triples $(A = \mathcal{O}(X),W,G)$ as long as the crepant resolution of $X/G$ exists and the lifting superpotential $W$ has good Hodge theoretical properties (see \cite{LW19} for a recent discussion on this model and references therein). The second is to establish the correspondence between \textsf{VSHS}'s via  crepant resolutions and related mirror symmetry. This involves a combination of LG/CY correspondence and mirror symmetry. Thirdly, there is a categorical approach to the orbifold LG models, which is called the equivariant matrix factorization (see, for example, \cite{D11,PV12,S17}). It would be very interesting to compare the categorical deformation theory with our calculations. We hope to investigate these problems in future works.

\noindent \textbf{Acknowledgements.}
The authors would like to thank beneficial discussions with Xiang Tang,  Junwu Tu, and Zhengfang Wang. The work of S. Li is partially supported by grant 11801300 of NSFC  and grant Z180003 of Beijing Natural Science Foundation. The work of W. He is partially supported by Tsinghua Postdoc Grant 100410019.

\section{Preliminary}
In this note, $\mathds{C}$ is taken as the base field for convenience. For a $\mathds{Z}$ or $\mathds{Z}_2$-graded vector space $A$, we denote by $sA$ its suspension, where $\lr{sA}_k \coloneqq \lr{A[-1]}_k = A_{k-1}$. We use the Koszul sign convention and regard $s$ as a degree $1$ element. Given two graded vector spaces $A$ and $M$, the spaces of Hochschild (co)chains and compact type Hochschild (co)chains are defined as

\begin{align*}
C^\bullet (A,M) \coloneqq \prod_{p\geqslant 0}\hom(\lr{sA}^{\otimes p},sM)[1],\quad& C^\bullet_c (A,M) \coloneqq \bigoplus_{p\geqslant 0}\hom(\lr{sA}^{\otimes p},sM)[1],\\
C_\bullet (A,M) \coloneqq \bigoplus_{p\geqslant 0}sM\otimes\lr{sA}^{\otimes p}[1],\quad\qquad& C_\bullet^c (A,M) \coloneqq \prod_{p\geqslant 0}sM\otimes\lr{sA}^{\otimes p}[1].
\end{align*}

We will write $[a_1| \cdots | a_p]$ for an element in $(sA)^{\otimes p}$, $m[a_1|\cdots |a_p]$ an element in $M\otimes (sA)^{\otimes p}$ and $\phi[a_1|\cdots |a_p]$ the value of $\phi\in C^p(A,M)$ acting on $[a_1| \cdots | a_p]$.

\begin{rem}For each $\phi\in C^p(A,M)$, we can associate $\phi\circ s^{\otimes p}\in \hom(A^{\otimes p},M)$ as
\[
\phi\circ s^{\otimes p}(a_1\otimes\cdots\otimes a_p)\coloneqq (-1)^{\sum\limits_{k=1}^{p-1}(p-k)|a_k|} \phi[a_1|\cdots|a_p].
\]
This fixes our sign conventions for Hochschild (co)chains.
\end{rem}

There are two different gradings for these (co)chains, the tensor grading and the internal grading, which are determined by the grading of $A$ and $M$. We denote by $\abs{\cdot}$ the internal grading. For example, for a homogeneous (with respect to both gradings) cochain $\phi\in C^p(A,M)$,

\begin{equation}
\abs{\phi} = \abs{\phi [a_1|\cdots| a_p]}  - \abs{a_1} - \cdots - \abs{a_p}-p.
\end{equation}

In \cite{VG95}, Gerstenhaber introduced the brace structure by higher operations, the braces on (compact type) Hochschild cochians. For homogeneous $\phi\in C^p(A,B)$ and $\phi_k\in C^{p_k}(B,C)$, we can define for $m = p+p_1+\cdots +p_n-n$,
\begin{align}
&\phi\{\phi_1,\cdots,\phi_n\}[a_1|\cdots |a_m]\nonumber\\ \coloneqq &(-1)^{\sum\limits_{k=1}^n\epsilon_{i_k}(\abs{\phi_k}+1)}\phi[a_1|\cdots |a_{i_k}|\phi_k[a_{i_{k}+1}|\cdots|a_{i_k+p_k}]|\cdots|a_m],
\end{align}
where
\[
\epsilon_i = \sum_{j=1}^i |a_j|-i.
\]
Notice that there is a one-shifted Lie algebraic structure on (compact type) Hochschild cochains $C^\bullet(A,A)$ (or $C^\bullet_c(A,A)$). It is defined \cite{G63} as the commutator of Gerstenhaber product (the brace operation with only one input),

\begin{equation}
\{\phi_1,\phi_2\} \coloneqq \phi_1\{\phi_2\}-(-1)^{(\abs{\phi_1}+1)(\abs{\phi_2}+1)}\phi_2\{\phi_1\}.
\end{equation}

In this note, we will work with $2$-dimensional orbifold Landau-Ginzburg models $(A_W,G)$. Here $A_W$ is denoted for a curved algebra $(A,W)$, where $A = \mathds{C}[x,y]$ and $W$ is an invertible polynomial. $G$ with the identity $\e$ is a finite group acting diagonally on $\mathds{C}^2$, which can be extended to an equivariant action on $A$. $W$ is asked to be $G$-invariant. (See \cite{K09,FJR13,HLL18} for details.) We will regard $A_W$ as a $\mathds{Z}/2\mathds{Z}$-graded \ainfty algebra concentrated in degree zero with $b_0 = -W$, $b_2[a_1|a_2] = (-1)^{|a_1|}a_1a_2=a_1a_2$ and $b_i = 0, \forall i\neq 0,2$. 
Similarly, the $G$-twisted curved algebra $A_W[G]$ is also regarded as a curved algebra on $A\otimes_{\mathds{C}} \mathds{C}[G]$ with $b_0 = -W\e$ and $b_2[a_1g_1|a_2g_2] = a_1{}^{g_1}a_2g_1g_2$ (Thus, $A_W[G] = A[G]_{W\e}$).

For an \ainfty algebra $A$, we can define boundary operators on the
(compact type) Hochschild (co)chains as follows. For $a_0[a_1|\cdots|a_p]\in C_p(A,A)$ (or $C_p^c(A,A)$),
\begin{align}
\p_{b}(a_0[a_1|\cdots |a_p]) \coloneqq &\sum_{l =1}^{p+1}\sum_{k=p+1-l}^p (-1)^{\varepsilon_k(\varepsilon_p-\varepsilon_k)} b_l[a_{k+1}|\cdots |a_0|\cdots][a_{k+l-p}|\cdots |a_k]\nonumber\\
&+ \sum_{l=0}^{p}\sum_{k=0}^{p-l}(-1)^{\varepsilon_k}a_0[\cdots|b_l[a_{k+1}|\cdots |a_{k+l}]| \cdots|a_p],
\end{align}
where
\begin{equation}
\varepsilon_k \coloneqq \abs{a_0}+\cdots + \abs{a_k} +k+1,
\end{equation}
and for $\phi \in C^p(A,A)$ (or $C^p_c(A,A)$),
\begin{equation}
\delta_b(\phi) \coloneqq \{b,\phi\}.
\end{equation}
In our cases,
\begin{equation}
\p_{b} = \p_{b_0}+\p_{b_2},
\end{equation}
with
\begin{equation}
\begin{dcases}
\p_{b_2}(a_0[a_1|\cdots |a_p]) \coloneqq &a_{0}a_1[a_2|\cdots |a_p]+(-1)^{p} a_pa_0[a_1|\cdots |a_{p-1}]\nonumber\\
&+ \sum_{k=0}^{p-2}(-1)^{k+1}a_0[a_1|\cdots|a_{k+1}a_{k+2}|\cdots|a_p],\\
\p_{b_0}(a_0[a_1|\cdots |a_p]) \coloneqq &\sum_{k=0}^{p-1}(-1)^{k}a_0[a_1|\cdots|a_k|W| a_{k+1}|\cdots|a_p].
\end{dcases}
\end{equation}
and
\begin{equation}
\delta_b\phi = \delta_{b_0}\phi+\delta_{b_2}\phi \coloneqq (-1)^{p-1}\phi\{W\} + \{b_2,\phi\}.
\end{equation}
The (compact type) Hochschild homology and cohomology are defined as the homology and cohomology of the (compact type) chains and cochains with differentials $\delta_b$ and $\p_b$ respectively.

While $A$ is augmented, we may consider the reduced Hochschild (co)chains defined on $\bar{A} = A/\mathds{C}$ (see \cite{L13} for details).

\section{Deformation Theory}

As we have shown in \cite{HLL18}, we can define higher operations on the $G$-twisted version of (compact type) Hochschild cochains. Thus, there is a Gerstenhaber algebra structure on
$\hh^\bullet_c(A_W,A_W[G])^G$, which is isomorphic to $\hh^\bullet_c(A_W[G],A_W[G])$.

\begin{thm}\label{thm: non-obs}
Consider orbifold LG B-models $(A,W,G)$ with $A=\mathds{C}[x,y]$, $W$ invertible and finite $G\subset \mathrm{SL}(2,\mathds{C})$ (the Calabi-Yau condition) acting diagonally on $\mathds{C}^2$. Then the shifted dgLa (differential graded Lie algebra)
\[(C^\bullet_c(A,A[G])^G,\delta_b, \{\,,\,\})\]
is homological abelian.
\end{thm}

\begin{proof}
Use the cochain version of the explicit homotopy retraction constructed in appendix \ref{appendix: A},
\[\begin{tikzpicture}
\matrix (a) [matrix of math nodes, row sep=3.5em, column sep=4em, text height=2ex, text depth=0.9ex]
{\big(C^{\bullet}_c(\bar{A},Ag)^G,\delta_b\big) & \big(\jac(W_g)^G[-l_g],\mathsf{0}\big),\\};
\path[-latex,semithick]
(a-1-1) edge (a-1-2);
\path[latex-,semithick]
(a-1-1.-5) edge (a-1-2.-176);
\draw[-latex,semithick]
(a-1-1.-170)  arc (-40:-320:11pt) -- ++ (.8pt,-.8pt);
\end{tikzpicture}\]
where $W_g \coloneqq W|_{\fix(g)}$ and
\[l_g = \begin{dcases}
0,g=\e,\\
2,g\neq \e.
\end{dcases}\]
By homotopy transfer theorem, we can define a shifted \linfty structure on the later, such that there exists a quasi-isomorphism between shifted \linfty algebras,
\[\big(C^\bullet_c(\bar{A},A[G])^G,\delta_b,\{\,,\,\}\big)\simeq \big(\jac(W,G),\mathsf{0},\ell_2+\ell_3+\cdots\big),\]
where
\[\jac(W,G) \coloneqq \bigoplus_{g\in G}\jac(W_g)^G[-l_g].\]
Notice that the degrees of $\ell_k$ are all odd, while $\jac(W,G)$ is concentrated in even degrees. Hence, all of those $\ell_k$'s are zero and $\big(C^\bullet_c(A,A[G])^G,\delta_b,\{\,,\,\}\big)$ is homological abelian.
\end{proof}

We can solve the Maurer-Cartan equation
\begin{equation}
\delta_b(\phi) +\f{1}{2}\{\phi,\phi\}=0,
\end{equation}
by the quasi-isomorphism defined above with homotopy transfer. With respect to the tensor grading, the homotopies on Hochschild and Koszul cochains are all of degree $-1$ , and the homotopy on polyvector fields shall not give terms of degree greater than $2$. Therefore, Maurer-Cartan elements are in the form of $\phi=\phi_0+\phi_2\in C^{\text{even}}_c(A_W,A_W[G])^G$, where $\phi_0$ gives a deformation of $b_0$ and $\phi_2$ gives a deformation of $b_2$. The miniversal deformation space of $A_W[G]$ is denoted by $\mathrm{Def}(A_W,G)$.

\begin{coro}
Under the same assumption as Theorem \ref{thm: non-obs}, a basis of $\hh^\bullet_c(A_W,A_W[G])^G$ will give a parametrization of a formal neighbourhood of the origin in $\mathrm{Def}(A_W,G)$.
\end{coro}

Notice that there is a decomposition,
\[\hh^\bullet_c(A_W,A_W[G])^G = \Big(\bigoplus_{g \neq \e }\hh^\bullet_c(A_W,A_Wg)\Big)^G\oplus \hh^\bullet_c(A_W,A_W\e)^G.\]
We call the first summand the twisted sector and the second the untwisted sector.

We can identify the formal neighbourhood of the origin in $\mathrm{Def}(A_W,G)$ with $\mathrm{Spec}(\mathds{C}[[\bm{\tau},\bm{s}]])$, where variables $\bm{\tau}$ parameterize deformations in the untwisted sector and $\bm{s}$ parameterize deformations in the twisted sectors.  Let $\mathcal{A}(\bm{\tau},\bm{s})$ denote the corresponding deformed curved algebra, with $b(\bm{\tau},\bm{s})$ the deformed \ainfty-products. Miniversality says that the following Kodaira-Spencer map is an isomorphism
\begin{align*}\ks\colon \text{Span}_{\mathds{C}}\Big\{{\pd{}{\bm{\tau}}}, {\pd{}{\bm{s}}}\Big\} &\to \hh^\bullet_{c}(A_W,A_W[G])^G,\\
v\quad\quad&\mapsto  \big[v(b(\bm{\tau},\bm{s}))|_{\bm{\tau}=\bm{s}=\bm{0}}\big].
\end{align*}

Also notice that the first order deformation along the untwisted sector deforms the superpotential, while that along the twisted sector deform the product of the semi-direct product polynomial ring $\mathds{C}[x,y][G]$.

\begin{rem}In \cite{N98}, Nadaud introduced three different forms of ``$q$-Moyal products'' as $q$-deformations of $\mathds{C}[x,y]$. By choosing different homotopy retractions between the Hochschild cochains and Koszul cochains, one can recover these $q$-Moyal products
and his rigidity indicates that our deformation is in fact somewhat canonical. One should also notice that our deformation is equivalent to that constructed by Halbout, Oudom and Tang in \cite{HOT11}. Especially, their modified superpotential is the same as our deformed $b_0$ via homotopy transfer.
\end{rem}

E. Getzler \cite{G93} introduced higher operations $\bm{b}$ and $\bm{B}$ on $C^\bullet(A,A)$ with values in $\mathrm{End}(C_\bullet(\bar{A},A))$ for any \ainfty algebra $A$, as extensions of Hochschild differential $\p_{b}$ and Connes operator $B$. Here $\bar{A} = \mathds{C}[x,y]/\mathds{C}$. For homogeneous $\phi_1,\cdots,\phi_n\in C^\bullet(A,A)$ and $a_0[a_1|\cdots |a_m]\in C_m(\bar{A},A)$, he defined for $n\geqslant 1$,
\begin{align}
&\bm{b}\{\phi_1,\cdots,\phi_n\}(a_0[a_1|\cdots |a_m])\nonumber\\
\coloneqq &\!\! \sum_{\bm{J}\in\mathcal{J}}\sum_{l\geqslant 1} (-1)^{\eta_{\bm{J}}} b_l[a_{j_0+1}|\cdots|\phi_1[a_{j_1+1}|\cdots]|\cdots |\phi_n[a_{j_n+1}|\cdots]|\cdots|a_0|\cdots][\cdots |a_{j_0}],
\end{align}
where
\[\mathcal{J} = \left\{\bm{J} = (j_0,\cdots,j_n)\left| \begin{dcases}&m-(l-1)-\sum_{k=1}^n \abs{\phi_k}+n\leqslant j_0\leqslant j_1,j_n +\abs{\phi_n}\leqslant m,\\
&j_k+\abs{\phi_k}\leqslant j_{k+1}, \forall 1\leqslant k\leqslant n-1.\end{dcases}\right.\right\}, \]
and
\begin{align}
\eta_{\bm{J}} = \varepsilon_{j_0}(\varepsilon_{m}-1)+\sum_{k=1}^n (\abs{\phi_k}-1)(\varepsilon_{j_k}-\varepsilon_{j_0}),
\end{align}
and
\begin{align}
&\bm{B}\{\phi_1,\cdots,\phi_n\}(a_0[a_1|\cdots |a_m])\nonumber\\
\coloneqq &\!\! \sum_{\bm{J}\in\mathcal{J}}\sum_{l\geqslant 1} (-1)^{\eta_{\bm{J}}} 1[a_{j_0+1}|\cdots|\phi_1[a_{j_1+1}|\cdots]|\cdots |\phi_n[a_{j_n+1}|\cdots]|\cdots|a_0|\cdots |a_{j_0}].
\end{align}
For $n=0$, $\bm{b}\{\}\coloneqq \p_{b}$ and $\bm{B}\{\} \coloneqq B$. Using $\bm{b}$ and $\bm{B}$, he defined the Getzler-Gauss-Manin system for formal deformations of an \ainfty algebra.

We will first extend his constructions to the case of $G$-twisted chains $C_\bullet^c(\bar{A},A[G])_G$, which is much smaller than the space of reduced chains $C_\bullet^c(\overline{A[G]},A[G])$. This will greatly simplify the calculation of the connection in practice.

Consider chain maps $\Psi_*$ and $\Gamma_*$ between $C_\bullet^c(A,A[G])_G $ and $ C_\bullet^c(A[G],A[G])$ defined as follows. For $a_0g_0[a_1g_1|\cdots |a_pg_p]\in C_p(A[G],A[G])$ and $a_0g_0[a_1|\cdots|a_p]\in C_p(A,A[G])$,
\begin{align}
\Gamma_*\circ\pi(a_0g_0[a_1|\cdots|a_p])& \coloneqq \f{1}{\abs{G}}\sum_{g\in G} {}^g a_0gg_0g^{-1}[{}^g a_1 e|\cdots | {}^g a_p e],\\
\Psi_*(a_0g_0[a_1g_1|\cdots|a_pg_p]) & \coloneqq \pi\big({}^{g_1\cdots g_p}a_0g_1\cdots g_pg_0[a_1|{}^{g_1}a_2|\cdots|{}^{g_1\cdots g_{p-1}}a_p]\big).
\end{align}
Here, $\pi \colon C_\bullet(A,A[G]) \to C_\bullet(A,A[G])_G$ is the natural projection. One can easily check that $\Gamma_*$ is well-defined.

A Getzler-Gauss-Manin connection is a connection defined in terms of the mixed complex $(C_\bullet, \p ,B)$ of a deformed \ainfty algebra, for example, on
\[\big(C_\bullet^c(\overline{\mathcal{A}(\bm{\tau},\bm{s})},\mathcal{A}(\bm{\tau},\bm{s})),\p_{b(\bm{\tau},\bm{s})},B\big) = \big(C_\bullet^c(\overline{A[G]},A[G])[[\bm{\tau},\bm{s}]],\p_{b(\bm{\tau},\bm{s})},B\big).\]
On twisted chains, we can also define such a mixed complex by defining
\begin{align}
\tilde{\p}_{b(\bm{\tau},\bm{s})}& \coloneqq \Psi_*\circ \p_{b(\bm{\tau},\bm{s})}\circ \Gamma_*,\\
\tilde{B} &\coloneqq \Psi_* \circ B \circ \Gamma_*.
\end{align}
More explicitly, if we write $b(\bm{\tau},\bm{s})\in C^\bullet(A, A[G])$ as
\[ b(\bm{\tau},\bm{s}) = \sum_{h\in G} (b^h_2 h-W^hh),\]
with $b^h_2\in C^2(A,A)[\bm{s}]$ and $W^h\in A[\bm{\tau},\bm{s}]$, we can write
\[\tilde{\p}_{b(\bm{\tau},\bm{s})} = \tilde{\p}_{b_2(\bm{\tau},\bm{s})}+\tilde{\p}_{b_0(\bm{\tau},\bm{s})},\]
where
\begin{align}
&\tilde{\p}_{b_2(\bm{\tau},\bm{s})}\circ \pi(a_0g_0[a_1|\cdots|a_p]) \nonumber\\
=& \sum_{h\in G}\pi\bigg(b_2^h[a_0|{}^{g_0}a_1]hg_0[a_2|\cdots|a_p]+(-1)^p b_2^h[a_p|a_0]hg_0[a_1|\cdots|a_{p-1}]\nonumber\\
&+ \sum_{k=0}^{p-2}(-1)^{k+1}{}^ha_0hg_0[a_1|\cdots|b^h_2[a_{k+1}|a_{k+2}]|{}^ha_{k+3}|\cdots|{}^h a_p]\bigg),
\end{align}
and
\begin{align}
\tilde{\p}_{b_0(\bm{\tau},\bm{s})}\circ \pi(a_0g_0[a_1|\cdots|a_p])=\sum_{h\in G}\pi\bigg(\sum_{k=0}^{p-2}(-1)^{k}{}^ha_0hg_0[\cdots|a_k|W^h|{}^ha_{k+1}|\cdots|{}^h a_p]\bigg).
\end{align}
Similarly,
\begin{align}
\tilde{B}\circ \pi(a_0g_0[a_1|\cdots|a_p]) \coloneqq \pi\bigg(\sum_{k=0}^p(-1)^{kp}1g_0[a_k|\cdots|a_p|a_0|{}^{g_0}a_1|\cdots|{}^{g_0}a_{k-1}]\bigg).
\end{align}
One can directly check that
\[\big(C_\bullet^c(\bar{A},A[G])_G[[\bm{\tau},\bm{s}]],\tilde{\p}_{b(\bm{\tau},\bm{s})},\tilde{B}\big)\]
is a mixed complex. Furthermore, $\Psi_*$ will induce a morphism between mixed complexes,
\[\Psi_*\colon \big(C_\bullet^c(\overline{A[G]},A[G])[[\bm{\tau},\bm{s}]],\p_{b(\bm{\tau},\bm{s})},B\big) \to \big(C_\bullet^c(\bar{A},A[G])_G[[\bm{\tau},\bm{s}]],\tilde{\p}_{b(\bm{\tau},\bm{s})},\tilde{B}\big).\]

\begin{lem}\label{lem: chain map compatiablity}
$\tilde{\p}_{b(\bm{\tau},\bm{s})}$ and $\tilde{B}$ defined above can be extended to higher operations $\tilde{\bm{b}}(\bm{\tau},\bm{s})$ and $\tilde{\bm{B}}$ on $C^\bullet_c(A,A[G])^G[[\bm{\tau},\bm{s}]]$ with values in $\mathrm{End}(C^c_\bullet(\bar{A},A[G])_G[[\bm{\tau},\bm{s}]]$, such that these higher operations are also compatible with $\Psi_*$.
\end{lem}

\begin{proof}
Similar as above, for homogeneous $\phi_1,\cdots,\phi_n\in C^\bullet(A,A[G])^G$, we define
\begin{align}
\tilde{\bm{b}}(\bm{\tau},\bm{s})\{\phi_1,\cdots,\phi_n\}\coloneqq &  \Psi_*\circ \bm{b}(\bm{\tau},\bm{s})\{\Psi^*(\phi_1),\cdots,\Psi^*(\phi_n)\}\circ \Gamma_*,\\
\tilde{\bm{B}}\{\phi_1,\cdots,\phi_n\}\coloneqq &  \Psi_*\circ \bm{B}\{\Psi^*(\phi_1),\cdots,\Psi^*(\phi_n)\}\circ \Gamma_*.
\end{align}
Here, $\bm{b}(\bm{\tau},\bm{s})$ and $\bm{B}$ are the higher operations on $C^\bullet_c(A[G],A[G])[[\bm{\tau},\bm{s}]]$ with values in $\mathrm{End}(C^c_\bullet(\overline{A[G]},A[G])[[\bm{\tau},\bm{s}]]$ extending $\p_{b(\bm{\tau},\bm{s})}$ and $B$ and $\Psi^*$ is the cochain maps we constructed in \cite{HLL18} such that
\begin{equation}
\Psi^*(\phi)[a_1g_1|\cdots|a_pg_p] = \phi[a_1|{}^{g_1}a_2|\cdots |{}^{g_1\cdots g_{p_1}}a_p]g_1\cdots g_p.
\end{equation}
One can also directly check that
\begin{align*}
\tilde{\bm{b}}(\bm{\tau},\bm{s})\{\phi_1,\cdots,\phi_n\}\circ \Psi_* =& \Psi_* \circ \bm{b}(\bm{\tau},\bm{s})\{\Psi^*(\phi_1),\cdots,\Psi^*(\phi_n)\},\\
\tilde{\bm{B}}\{\phi_1,\cdots,\phi_n\} \circ \Phi_* =& \Psi_* \circ \bm{B}\{\Psi^*(\phi_1),\cdots,\Psi^*(\phi_n)\}.
\end{align*}
\end{proof}

By Getzler's computation \cite{G93}, there is a connection flat up to homotopy defined as
\begin{align}
\nabla \colon  \mathds{C}[[\bm{\tau},\bm{s}]]\Big[\pd{}{\bm{\tau}}, \pd{}{\bm{s}}\Big] &\to \mathrm{End}_{\mathds{C}}\Big(C^c_\bullet(\overline{A[G]},A[G])[[\bm{\tau},\bm{s}]]((u))\Big),\nonumber\\
\nabla_{v} &\coloneqq v-\f{1}{u}\bm{b}(\bm{\tau},\bm{s})\{v(b(\bm{\tau},\bm{s}))\} - \bm{B}\{v(b(\bm{\tau},\bm{s}))\},
\end{align}
which induces a flat connection on
\[\hp^c_\bullet(\mathcal{A}(\bm{\tau},\bm{s})) \coloneqq \mathrm{H}_\bullet\Big(C^c_\bullet(\overline{A[G]},A[G])[[\bm{\tau},\bm{s}]]((u)),\p_{b(\bm{\tau},\bm{s})}+uB\Big).\]
$\nabla$ is flat up to homotopy on chain level, which means that $\nabla$ is flat on homologies. Thus, there is a flat connection
\begin{align}
\nabla \colon  \mathds{C}[[\bm{\tau},\bm{s}]]\Big[\pd{}{\bm{\tau}}, \pd{}{\bm{s}}\Big]  &\to \mathrm{End}_{\mathds{C}}\Big(\hp^c_\bullet(\mathcal{A}(\bm{\tau},\bm{s}))\Big),\nonumber\\
v &\mapsto[\nabla_v(-)].
\end{align}
We can give a similar definition for $\nabla$ on $G$-twisted chains just by replacing $\bm{b}$ and $\bm{B}$ with $\tilde{\bm{b}}$ and $\tilde{\bm{B}}$.  This will also induce a flat connection on
\[\mathrm{H}_\bullet \big(C_\bullet^c(\bar{A},A[G])_G[[\bm{\tau},\bm{s}]]((u)), \tilde{\p}_{b(\bm{\tau},\bm{s})}+u\tilde{B}\big),\]
by the same reason. Lemma \ref{lem: chain map compatiablity} and the fact that $\Psi_*$ is a quasi-isomorphism \cite{HLL18} explain why these two constructions define the same flat connection on the periodic cyclic homology compatible with the Hodge filtration.

\section{An Example: \texorpdfstring{$ (A_{2n-1},\mathds{Z}_2)$}) Cases}
In this section, we will write down the Getzler-Gauss-Manin system on the miniversal deformation of $A_{2n-1}$ type orbifold explicitly. Here $W= x^{2n}+y^2$ ($n \geqslant 2$), the orbifold group is $G=\mathds{Z}_2$, whose generator ${\sigma}$ acts on $x, y$ by ${}^\sigma x = -x, {}^\sigma y=-y$. The result turns out to coincide with the Gauss-Manin system on the miniversal deformation of $D_{n+1}$ singularity.  This establishes an example of crepant resolution conjecture for LG B-models over miniversal deformations. In fact,  if we lift the superpotential $W$ to the crepant resolution of $\mathds{C}^2/\mathds{Z}_2$, which is the total space of $\mathcal{O}_{\bm{P}^1}(-2)$, it will have an isolated singularity of $D_{n+1}$ type on the exceptional $\bm{P}^1$.

\noindent \textbf{Computation of A-type orbifolds: $W= x^{2n}+y^2, G=\mathds{Z}_2$}

Since $\jac(W,G) = \jac(W)^G e\oplus \mathds{C}\sigma[-2]$, the formal neighbourhood of the origin in $\mathrm{Def}(A_W,G)$ can be parameterized as $\mathrm{Spec}(\mathds{C}[[\tau_0,\tau_1,\cdots,\tau_{n-1},s]])$ and the deformed curved algebra $\mathcal{A}(\bm{\tau},s)$ is $\mathds{C}[[\bm{\tau},s]][x,y]\otimes \mathds{C}[G]$ with $b(\bm{\tau},s)$ given by
\begin{equation}\begin{dcases}
  b_l(\bm{\tau},s) = 0, \text{ if }l\neq 0,2,\\
  b_2(\bm{\tau},s)  = b_2+ s\p_x^\sigma\p_y^\sigma,\\
   b_0(\bm{\tau},s) = -W_{\bm{\tau}} \coloneqq -x^{2n} -y^2 - \sum\limits_{k=0}^{n-1} \tau_k x^{2k}.
\end{dcases}\end{equation}
Here $\p_x^\sigma$ and $\p_y^\sigma$ are the quantum differential operators defined in \cite{HLL18}. As a Hochschild cochain,
\[\p^\sigma_x\p^\sigma_y[x^{a_1}y^{b_1}|x^{a_2}y^{b_2}] \coloneqq\begin{dcases} (-1)^{a_2}x^{a_1+a_2-1}y^{b_1+b_2-1}\sigma,& \text{ if } a_1,b_2 \text{ odd,}\\
0,&\text{ else.}
\end{dcases}
\]


Henceforth, we will write
\begin{equation}
b(\bm{\tau},s) = b_2 + s b_\sigma + b_0(\bm{\tau}),
\end{equation}
with $b_\sigma = \p^\sigma_x\p^\sigma_y$ and $b_0(\bm{\tau}) \coloneqq b_0(\bm{\tau},s)$. Hence, the deformed Hochschild differential is
\begin{equation}
\tilde{\p}_{b(\bm{\tau},s)} = \tilde{\p}_{b_2} + s \tilde{\p}_{b_\sigma} + \tilde{\p}_{b_0(\bm{\tau})}.
\end{equation}

In these cases, $\mathrm{HP}^c_\bullet(\mathcal{A}(\bm{\tau},s))$ equals to
\begin{equation}
\lr{\jac(W)_G[2] \oplus \jac(W_\sigma)}[[\bm{\tau},s]]((u)) = \Big(\bigoplus_{k=0}^{n-1}\mathds{C} [x^{2k}][2] \oplus \mathds{C}[1\sigma]\Big)[[\bm{\tau},s]]((u)),
\end{equation}
as a $\mathds{Z}_2$-graded $\mathds{C}[[\bm{\tau},s]]((u))$-module with a flat connection $\nabla$ we have defined in the last section.

Using perturbed homotopy retraction \cite{C04} constructed by homotopy retractions defined in the appendix, the chain representations for $[x^{2k}][2]\in \jac(W)_G[2]$ and $[1_\sigma]\in \jac(W_\sigma)$ can be written in the following forms,
\[\begin{dcases}
\ha_{2k} &= \ha^{(2)}_{2k} + \ha^{(4)}_{2k} +\cdots,\\
\hb & = \hb^{(0)} + \hb^{(2)} +\cdots,
\end{dcases}\]
where $\ha^{(p)}_{k}, \hb^{(p)}\in C_{p}(\bar{A},A[G])_G[\bm{\tau},s,u]$ are defined as follows. For $\forall 0\leqslant k\leqslant n-1$, take $\ha^{(2)}_{2k} = x^{2k}[x|y]-x^{2k}[y|x]$ and $\hb^{(0)} = 1\sigma$; and for $l\geqslant 2$, we define
  \begin{align}
  \ha^{(2l)}_{2k} = -\sum_{i\geqslant 0}\big(-s(\ho_C+\Phi \ho_K\Upsilon)\tilde{\p}_{b_\sigma}\big)^i(\ho_C+\Phi \ho_K\Upsilon)(\tilde{\p}_{b_0(\bm{\tau})}+u\tilde{B})\ha^{(2l-2)}_{2k}.
  \end{align}
  Similarly, $\forall l\geqslant 1$, we define
  \begin{align}
  \hb^{(2l)} = -\sum_{i\geqslant 0}\big(-s(\ho_C+\Phi \ho_K\Upsilon)\tilde{\p}_{b_\sigma}\big)^i(\ho_C+\Phi \ho_K\Upsilon)(\tilde{\p}_{b_0(\bm{\tau})}+u\tilde{B})\hb^{(2l-2)}.
  \end{align}

The lower degree terms of $\ha_{2k}$ and $\hb$ which will be used in further calculation are given by
\begin{small}
\begin{align}
\ha_{2k}^{(2)} = & x^{2k}[x\wedge y],\label{eq: alpha2}\\
\ha_{2k}^{(4)} = & \sum_{j=1}^{n} \sum_{i=0}^{2j-2} \tau_{j} \big(x^{i+2k}[x|x^{2j-i-1}|x\wedge y]+x^{i+2k}[x\wedge y|x^{2j-i-1}|x]\big)\label{eq: alpha4}\\
&+x^{2k}[x\wedge y|y|y] + x^{2k}[y|y|x\wedge y] +u\sum_{i=0}^{2k-2} \big(x^i[x\wedge y|x^{2k-i-1}|x]+x^i[x|x^{2k-i-1}|x\wedge y]\nonumber\\
\hb^{(0)} = & 1\sigma,\label{eq: beta0}\\
\hb^{(2)} = & \sum_{j=1}^{n} \sum_{i=0}^{2j-2}\tau_{j} x^i\sigma[x^{2j-i-1}|x] +1\sigma[y|y],\label{eq: beta2}\\
\hb^{(4)} = & -u\ho_C\tilde{B}(\hb^{(2)}) + \sum_{j,j'=1}^{n} \sum_{i=0}^{2j-2}\sum_{i'=0}^{2j'-2}\tau_{j} \tau_{j'}x^{i+i'}\sigma[x^{2j-i-1}|x|x^{2j'-i'-1}|x]\label{eq: beta4}\\
& +\sum_{j=1}^n\sum_{i=0}^{2j-2} \tau_{j} \big(x^i\sigma[x^{2j-i-1}|x|y|y] - x^i\sigma[x^{2j-i-1}|y|x\wedge y]\nonumber\\
&+ \sum_{j=1}^n\sum_{i=0}^{2j-2} \tau_{j} \big(x^i\sigma[y|x^{2j-i-1}|x\wedge y]+x^i\sigma[y|y|x^{2j-i-1}|x]\big)\nonumber\\
& + 1\sigma[y|y|y|y].\nonumber
\end{align}
\end{small}
Here, for the sake of simplicity, we denote $\tau_{n}=1$ and $\cdots|x\wedge y|\cdots = \cdots|x|y|\cdots-\cdots|y|x|\cdots$.

\noindent \textbf{Computation of D-type: $W=z^n+zw^2, G=\{1\}$}

For $G$ trivial, D. Shklyarov had shown in \cite{S16} that the Gauss-Manin system via similar non-commutative methods is equivalent to that given by Saito's singularity theory \cite{S83}. Denote by $B_{\hat{W}} = \mathds{C}[z,w]_{z^n+zw^2}$ a curved polynomial algebra with a curvature $\hat{W} = z^n+ zw^2$ and consider its deformation as
\[\mathcal{B}(\bm{\tau},s) \coloneqq B_{\hat{W}(\bm{\tau},s)}, \text{ with }\hat{W}(\bm{\tau},s) \coloneqq z^n +\sum_{j=0}^{n-1} \tau_j z^j + zw^2 -sw.\]
$\mathrm{HP}^c_\bullet(\mathcal{B}(\bm{\tau},s))$ equals to $\Big(\bigoplus\limits_{k=0}^{n-1}\mathds{C} [z^{k}][2] \oplus \mathds{C}[w][2]\Big)[[\bm{\tau},s]]((u))$ while regarded as a $\mathds{Z}_2$-graded $\mathds{C}[[\bm{\tau},s]]((u))$-module with a flat connection $\nabla$.
Similar as above, we can also find chain representations for $[z^k]$ and $[w]$ in the following forms,
\[\begin{dcases}
\hat{\ha}_{2k} &= \hat{\ha}^{(2)}_{2k}+\hat{\ha}^{(4)}_{2k}+\cdots,\\
\hat{\hb} &=  \hat{\hb}^{(2)}+\hat{\hb}^{(4)}+\cdots,
\end{dcases}\]
where $\hat{\ha}^{(p)}_k,\hat{\hb}^{(p)}\in C_p(B,B)[\bm{\tau},s,u]$ satisfies that
\begin{align}
\hat{\ha}_{2k}^{(2)} &=  z^{k} [z|w]-z^k[w|z], \forall 0\leqslant k\leqslant n,\\
\hat{\hb}^{(2)} & = w[z|w]-w[w|z].
\end{align}

Consider the bundle map $\Lambda$ on $\mathrm{Spec}(\mathds{C}[[\bm{\tau},s]])$, which maps sections of $\hp^c_\bullet(\mathcal{A}(\bm{\tau},s))$ to sections of $\hp^c_\bullet(\mathcal{B}(\bm{\tau},s))$ induced by $\ha_{2k} \mapsto \hat{\ha}_{2k}$ and $\hb\mapsto\hat{\hb}$.

\begin{thm}\label{main-thm} Viewed as bundles over $\mathrm{Spec}(\mathds{C}[[\bm{\tau},s]])$, the periodic cyclic homology $\hp^c_\bullet(\mathcal{A}(\bm{\tau},s))$ of the deformed curved algebra $\mathcal{A}(\bm{\tau},s)$ associated to the orbifold LG B-model $(\mathds{C}[x,y],x^{2n}+y^2,\mathds{Z}_2)$ and that of the deformed curved algebra $\mathcal{B}(\bm{\tau},s)$ associated to the LG B-model $(\mathds{C}[z,w],z^{2n}+zw^2)$ are isomorphic via the bundle map $\Lambda$ defined above. Furthermore, this isomorphism is compatible with the Getzler-Gauss-Manin connections on both bundles. To be explicit, for any $v\in \mathrm{Der}_{\mathds{C}}\mathds{C}[[\bm{\tau},s]]$ and any section $\theta$ of the bundle $\hp^c_\bullet(\mathcal{A}(\bm{\tau},s))$, we have
\begin{equation}\label{eq: connection compatibility}
\Lambda(\nabla_v(\theta)) = \nabla_v(\Lambda(\theta)).
\end{equation}
\end{thm}

\begin{proof}
We will proof this by direct calculation using tools coming from homotopy perturbations. Notice that we only need to show (\ref{eq: connection compatibility}) for $v = \pd{}{\tau_j}$ or $\pd{}{s}$ and $\theta = [\alpha_{2k}],0\leqslant k\leqslant {n-1}$ or $[\beta]$.
\begin{description}
  \item[Case 1] $v = \pd{}{s}$ and $\theta = [1\sigma]$. We have
  \begin{align*}
  [\nabla_{\pd{}{s}}(\hat{\hb})] = -\f{1}{u}\big(-\sum_{j=1}^n j\tau_j [\hat{\ha}_{2j-2}]\big).
  \end{align*}
  By (\ref{eq: beta0}), (\ref{eq: beta2}) and (\ref{eq: beta4}), we have
  \begin{align*}
  \nabla_{\pd{}{s}}(\hb) = &-\f{1}{u}\big(\sum_{j=1}^n \tau_j \sum_{i=0}^{2j-2} (-1)^{i+1}(x^i[x^{2j-i-1}|y]-x^i[y|x^{2j-i-1}])\big) + (\text{order } \geqslant 4 \text{ terms}),
  \end{align*}
  While acted on the above by the perturbed projection,
  \begin{align*}&\sum_{l\geqslant 0}p\Pi\Upsilon\Big(- (s\tilde{\p}_{b_\sigma}+\tilde{\p}_{b_0(\bm{\tau},s)-b_0}+u\tilde{B})\sum_{i\geqslant 0}\big(-(\ho_C+\Phi\ho_K\Upsilon)\tilde{\p}_{b_0}\big)^i\\&\qquad\quad (\ho_C+ \Phi\ho_K\Upsilon + \Phi\Theta\ho_\Omega\Pi\Upsilon)\Big)^l,
  \end{align*}
  the section represented by the above is given by
  \begin{align*}
  -\f{1}{u}\sum_{j=1}^n \tau_j\Big(\sum_{i=0}^{2j-2} (-1)^{i+1}(2j-i-1)\Big)[\ha_{2j-2}]= -\f{1}{u}\big(-\sum_{j=1}^n j\tau_j[\ha_{2j-2}]\big).
  \end{align*}
  Hence
  \begin{equation}
  \Lambda(\nabla_{\pd{}{s}}([\hb])) = \nabla_{\pd{}{s}}(\Lambda([\hb])).
  \end{equation}
  \item[Case 2] $v = \pd{}{s}$ and $\theta = [\ha_{2k}]$ for $0\leqslant k\leqslant n-1$.
      \[\nabla_{\pd{}{s}}([\hat{\ha}_{2k}]) = \begin{dcases}
      -\f{1}{u} [\hat{\beta}],&k=0,\\
      -\f{1}{u} \big(\f{s}{2}[\hat{\ha}_{2k-2}]\big),&0< k\leqslant n-1.
      \end{dcases}\]
      And similarly,
      \begin{align}\nabla_{\pd{}{s}}(\ha_{2k}) =&-\f{1}{u}\big(x^{2k}\sigma + x^{2k}\sigma[y|y]+\sum_{j=1}^n \sum_{i=0}^{2j-2} \tau_j(-1)^ix^{2k+i}\sigma[x|x^{2j-i-1}]\nonumber\\ &\qquad+ u\sum_{i=0}^{2k-2} (-1)^i x^i\sigma [x|x^{2k-i-1}]\big)
      + (\text{order }\geqslant 4 \text{ terms}),\end{align}
      so after acted by the perturbed projection
      \begin{align*}
      \nabla_{\pd{}{s}}([\ha_{2k}]) = \begin{dcases}
      -\f{1}{u} [\hb], &k=0,\\
      -\f{1}{u} \big(\f{s}{2}[\ha_{2k-2}]\big),&0< k\leqslant n-1.
      \end{dcases}
      \end{align*}
      Hence,
      \begin{equation}
      \Lambda(\nabla_{\pd{}{s}}([\ha_{2k}])) = \nabla_{\pd{}{s}}(\Lambda([\ha_{2k}])).
      \end{equation}
  \item[Case 3] $v= \pd{}{\tau_k}$ for $0\leqslant k\leqslant n-1$ and $\theta = [\hb]$. By
      \[[\nabla_{\pd{}{s}},\nabla_{\pd{}{\tau_k}}] =0,\]
      and Case 2, we have obviously
      \begin{equation}
      \Lambda(\nabla_{\pd{}{\tau_k}}([\hb])) = \nabla_{\pd{}{\tau_k}}(\Lambda([\hb])).
      \end{equation}
  \item[Case 4] $v = \pd{}{\tau_l}$ and $\theta = [\ha_{2k}]$ for $0\leqslant k,l\leqslant n-1$. $\nabla_{\pd{}{\tau_l}}(\Lambda([\ha_{2k}]))$ is given by the homology class of
      \[ -\f{1}{u}\big(-z^{k+l}[z\wedge w]\big) + (\text{order }\geqslant 4 \text{ terms}),\]
      and $\nabla_{\pd{}{\tau_l}}([x^{2k}])$ is given by the homology class of
      \[ -\f{1}{u}\big(-x^{2k+2l}[x\wedge y]\big) + (\text{order }\geqslant 4 \text{ terms}),\]
      so we can show the statement (\ref{eq: connection compatibility}) by induction on $k+l$. In cases $k+l\leqslant n-1$, (\ref{eq: connection compatibility}) is obvious and we can assume that (\ref{eq: connection compatibility}) holds for $k+l\leqslant m-1$ with $2n-2 \geqslant m\geqslant n$. Then for $k+l = m$, $\nabla_{\pd{}{\tau_l}}(\Lambda([\ha_{2k}]))$ equals to the homology class of
      \begin{align*}
      -\f{1}{u}\Big(\sum_{j=1}^{n-1} \f{j}{n}\tau_j\big(z^{m-n+j} [z\wedge w]\big)+\f{1}{2n}s \big(z^{m-n}w[z\wedge w]\big)+u\f{2m-2n+1}{2n}\big(z^{m-n}[z\wedge w]\big)\Big).
      \end{align*}
      And similarly, $\nabla_{\pd{}{\tau_l}}([\ha_{2k}])$ equals to the homology class of
      \begin{align*}-\f{1}{u}\Big(\sum_{j=1}^{n-1} \f{j}{n}\tau_j \big(x^{2m-2n+2j}[x\wedge y]\big) + \f{1}{2n}s x^{2m-2n}\sigma +u\f{2m-2n+1}{2n}\big(x^{2m-2n}[x\wedge y]\big)\Big),
      \end{align*}
      By the same calculation in the above cases and our assumption, (\ref{eq: connection compatibility}) holds in these cases.
\end{description}
\end{proof}

\begin{appendices}
\section{Constructions of the Homotopies}\label{appendix: A}
Since we are working on the $G$-twisted chains, a direct calculation on the '$G$-twisted version' of periodic homology of those deformed algebras is needed. This can be done by constructing an explicit special homotopy retraction.

Firstly, consider the Koszul chains,
\[K_\bullet(A,A[G]) \coloneqq \bigoplus_{p\geqslant 0}A[G]\otimes \mathds{C}[e_1,e_2],\]
with $G$-action given by
\begin{equation}
g.\big(ahe_{i_1} \cdots e_{i_p}\big) = {}^gaghg^{-1}{}^{g}e_{i_1}\cdots{}^ge_{i_p}.
\end{equation}
Here $e_i$ are the odd parameters with respect to $x_i$. On Koszul chains, we can define a differential call a Koszul differential as
\begin{equation}
\p_K(age_{\bm{I}}) \coloneqq \sum_{k=1}^p(-1)^{k-1} ({}^gx_{i_k}-x_{i_k})age_{\bm{I}\setminus\{i_k\}},
\end{equation}
where $\bm{I} = \{i_1 < \cdots < i_p\}\subseteq \{1,2\}$. In \cite{SW11}, Shepler and Witherspoon introduced two chain maps $\Phi$ and $\Upsilon$ and in \cite{HLL18}, we construct a homotopy $\ho_C$ such that $\Phi,\Upsilon $ are both $G$-equivariant and we have a special homotopy retraction
\[\begin{tikzpicture}
\matrix (a) [matrix of math nodes, row sep=3.5em, column sep=4em, text height=2ex, text depth=0.9ex]
{\big(C_{\bullet}(\bar{A},A[G])_G,\tilde{\p}_{b_2}\big) & \big( K_{\bullet}(A,A[G])_G,\p_K\big).\\};
\path[-latex,semithick]
(a-1-1) edge node[above,font=\scriptsize] {$\Upsilon$} (a-1-2);
\path[latex-,semithick]
(a-1-1.-3) edge node[below,font=\scriptsize] {$\Phi$} (a-1-2.-177);
\draw[-latex,semithick]
(a-1-1.-170)  arc (-40:-320:11pt) -- ++ (.8pt,-.8pt);
\node[left,outer sep=20pt,font=\scriptsize] at (a-1-1.-180) {$\ho_C$};
\end{tikzpicture}\]
Equivalently, $(\Phi,\Upsilon,\ho_C)$ satisfies that
\[
\begin{dcases}
\Upsilon \circ \Phi = \id, \id-\Phi\circ \Upsilon = [\tilde{\p}_{b_2},\ho_C];\\
\ho_C\circ \ho_C = 0, \ho_C \circ \Phi =0, \Upsilon\circ \ho_C=0.
\end{dcases}
\]

Secondly, we can also construct a special homotopy retraction
\[\begin{tikzpicture}
\matrix (a) [matrix of math nodes, row sep=3.5em, column sep=4em, text height=2ex, text depth=0.9ex]
{\big( K_{\bullet}(A,A[G])_G,\p_K\big) & \big( \bigoplus\limits_{g\in G}\Omega^\bullet(\fix(g))_G,0\big).\\};
\path[-latex,semithick]
(a-1-1) edge node[above,font=\scriptsize] {$\Pi$} (a-1-2);
\path[latex-,semithick]
(a-1-1.-3) edge node[below,font=\scriptsize] {$\Theta$} (a-1-2.-177);
\draw[-latex,semithick]
(a-1-1.-170)  arc (-40:-320:11pt) -- ++ (.8pt,-.8pt);
\node[left,outer sep=20pt,font=\scriptsize] at (a-1-1.-180) {$\ho_K$};
\end{tikzpicture}\]
Here, chain maps $\Pi$ and $\Theta$ between
$\big( K_{\bullet}(A,A[G]),\p_K\big)$ and  $\big( \bigoplus\limits_{g\in G}\Omega^\bullet(\fix(g)),0\big)$
are defined as
\[
\Pi (x_1^{\gamma_1}x_2^{\gamma_2}g e_{\bm{I}}) \coloneqq
\begin{dcases} \lr{x_1^{\gamma_1}x_2^{\gamma_2}}|_{\fix(g)}\df x_{\bm{I}},&\text{if }\bm{I}_g \cap\bm{I}=\emptyset,\\
 0,\quad&\text{else,}
\end{dcases}
\]
where $\bm{I}_g = \{i =1,2 \mid \lambda_i\neq 1\}$, and for a differential form $x_1^{\gamma_1}x_2^{\gamma_2}\df x_{\bm{I}}\in \Omega^{\bullet}(\fix(g))$,
\[
\Theta(x_1^{\gamma_1}x_2^{\gamma_2}\df x_{\bm{I}}) \coloneqq x_1^{\gamma_1}x_2^{\gamma_2}ge_{\bm{I}}.
\]
We can define the weight of a Koszul chain as
\begin{align}
\mathsf{wt}(x_1^{\gamma_1}x_2^{\gamma_2} g e_{\bm{I}}) = \sum_{k,\lambda_k\neq 1} \gamma_k+\sum_{i\in \bm{I},\lambda_i\neq 1} 1.
\end{align}
Then for a chain $\kappa$ with $\mathsf{wt}(\kappa) \neq 0$, we can define the homotopy $\ho_K$ as
\begin{equation}
\ho_K \colon \kappa = x_1^{\gamma_1}x_2^{\gamma_2} g e_{\bm{I}} \mapsto \sum_{i\in \bm{I},\lambda_{i}\neq 1} \f{1}{\mathsf{wt}(\kappa)} \f{1}{\lambda_i-1}\f{\p }{\p x_i}\lr{x_1^{\gamma_1}x_2^{\gamma_2}} g e_i\wedge e_{\bm{I}},
\end{equation}
where ${}^g x_i = \lambda_i x_i$. If $\mathsf{wt}(\kappa) = 0$, we will ask $\ho_K(\kappa) = 0$.
Notice that we also have $\Pi,\Theta $ and $\ho_K$ are all $G$-equivariant and they give a special homotopy retraction. The former is obvious and the later is by
\begin{enumerate}
  \item $\Pi \circ \Theta = \id$,
  \item $\id - \Theta \circ \Pi = [\ho_K, \p_K]$, because for $\kappa = age_{\bm{I}}$ with $\mathrm{wt}(\kappa) \neq 0$, we have
  \begin{align*}
  [\ho_K,\p_K](\kappa) = & \sum_{i\in\bm{I}\bigcap \bm{I}_g} \f{1}{\wt (\kappa)} \f{\p }{\p x_i}\lr{x_ia} ge_{\bm{I}}+ \sum_{i\in \bm{I}_g\setminus \bm{I}} \f{1}{\wt(\kappa)} x_i\f{\p }{\p x_i}\lr{a}ge_{\bm{I}}\\
   = & \kappa.
  \end{align*}
  and for $\kappa$ with $\wt(\kappa) =0 \iff \kappa = \Theta \circ \Pi(\kappa)$, we have  $[\ho_K,\p_K](\kappa)=0$.
  \item $\ho_K \circ \ho_K = 0$, $\ho_K \circ \Theta = 0$ and $\Pi \circ \ho_K = 0$ for obvious reasons.
\end{enumerate}

As a direct corollary, we have the following homotopy retraction,
\[\begin{tikzpicture}
\matrix (a) [matrix of math nodes, row sep=3.5em, column sep=4em, text height=2ex, text depth=0.9ex]
{\big( C_{\bullet}(\bar{A},A[G])_G,\tilde{\p}_{b_2}\big) & \big( \bigoplus\limits_{g\in G}\Omega^\bullet(\fix(g))_G,0\big).\\};
\path[-latex,semithick]
(a-1-1) edge node[above,font=\scriptsize] {$\Pi\circ\Upsilon$} (a-1-2);
\path[latex-,semithick]
(a-1-1.-3) edge node[below,font=\scriptsize] {$\Phi \circ \Theta$} (a-1-2.-177);
\draw[-latex,semithick]
(a-1-1.-170)  arc (-40:-320:11pt) -- ++ (.8pt,-.8pt);
\node[left,outer sep=20pt,font=\scriptsize] at (a-1-1.-180) {$\ho_C + \Phi \circ \ho_K \circ \Upsilon$};
\end{tikzpicture}\]
By taking the compact type, we can regard $\tilde{\p}_{b_0}$ as a small perturbation of $\tilde{\p}_{b_2}$ to get a perturbed special homotopy retraction. (See for example \cite{C04}.)

\begin{lem} The induced differential on $\Omega^\bullet(\fix(g))_G$ by the perturbation with respect to $\tilde{\p}_{b_0}$ is given by $\df W_g \wedge$.
\end{lem}

\begin{proof}By definition, the induced differential is given by
\[\Pi\Upsilon \circ\tilde{\p}_{b_0}\circ\Phi\Theta+ \sum_{i\geqslant 1}(-1)^i\Pi\Upsilon \circ \tilde{\p}_{b_0}\big((\ho_C+ \Phi\ho_K\Upsilon)\tilde{\p}_{b_0}\big)^i\circ \Phi\Theta.\]
Notice that
\begin{equation}
\Pi\Upsilon \circ\tilde{\p}_{b_0}\circ\Phi\Theta|_{\Omega^\bullet(\fix(g))_G} = \df W_g\wedge, \text{ and }\Pi\Upsilon \circ\tilde{\p}_{b_0} = (\df W_g\wedge) \circ\Pi\Upsilon \text{ on }\Omega^\bullet(\fix(g))_G,
\end{equation}
so the higher order terms in the differential on the $g$-sector $\Omega^\bullet(\fix(g))_G$ can be written as
\begin{align*}
&\sum_{i\geqslant 1}(-1)^i\Pi\Upsilon \circ \tilde{\p}_{b_0}\big((\ho_C+ \Phi\ho_K\Upsilon)\tilde{\p}_{b_0}\big)^i\circ \Phi\Theta\\
=& \sum_{i\geqslant 1}(-1)^i (\df W_g \wedge) \circ \Pi\Upsilon \big((\ho_C+ \Phi\ho_K\Upsilon)\tilde{\p}_{b_0}\big)^i\circ \Phi\Theta\\
=& 0.
\end{align*}
The last equality is because $\Upsilon \ho_c = 0$ and $\Pi \Upsilon\Phi \ho_K = \Pi \ho_K=0$.
\end{proof}

In summary, we have the following homotopy retraction,
\begin{equation}\label{eq: undeformed homotopy}\begin{tikzpicture}
\matrix (a) [matrix of math nodes, row sep=3.5em, column sep=4em, text height=2ex, text depth=0.9ex]
{\big( C_{\bullet}^c(\bar{A},A[G])_G,\tilde{\p}_{b}\big) & \big( \bigoplus\limits_{g\in G}\jac(W_g)_G[\overline{2-\abs{\bm{I}_g}}],0\big).\\};
\path[-latex,semithick]
(a-1-1) edge (a-1-2);
\path[latex-,semithick]
(a-1-1.-4) edge (a-1-2.-177);
\draw[-latex,semithick]
(a-1-1.-170)  arc (-40:-320:11pt) -- ++ (.8pt,-.8pt);
\end{tikzpicture}\end{equation}
Here the homotopy is given by
\begin{align}\label{eq: curved homotopy}
&\sum_{n\geqslant 0}\big(-(\ho_C+\Phi\ho_K\Upsilon)\tilde{\p}_{b_0}\big)^n(\ho_C+ \Phi\ho_K\Upsilon) + \sum_{m\geqslant 0}\big(-(\ho_C+\Phi\ho_K\Upsilon)\tilde{\p}_{b_0}\big)^m\Phi\Theta\ho_\Omega\Pi\Upsilon\nonumber\\
=& \sum_{n\geqslant 0}\big(-(\ho_C+\Phi\ho_K\Upsilon)\tilde{\p}_{b_0}\big)^n(\ho_C+ \Phi\ho_K\Upsilon + \Phi\Theta\ho_\Omega\Pi\Upsilon).
\end{align}
where $\ho_\Omega$ is some appropriate homotopy on $\bigoplus\limits_{g\in G}\Omega^\bullet(\fix(g))_G$, such that
\[\begin{tikzpicture}
\matrix (a) [matrix of math nodes, row sep=3.5em, column sep=4em, text height=2ex, text depth=0.9ex]
{\big( \bigoplus\limits_{g\in G}\Omega^\bullet(\fix(g))_G,\bigoplus\limits_{g\in G}\df W_g\big) & \big( \bigoplus\limits_{g\in G}\jac(W_g)_G[\overline{2-\abs{\bm{I}_g}}],0\big),\quad\\};
\path[-latex,semithick]
(a-1-1) edge node[above,font=\scriptsize] {$p$} (a-1-2);
\path[latex-,semithick]
(a-1-1.-3) edge node[below,font=\scriptsize] {$i$}(a-1-2.-177);
\draw[-latex,semithick]
(a-1-1.-172) ++ (-1pt,0) arc (-40:-320:11pt) --  ++ (.8pt,-.8pt);
\node[left,outer sep=20pt,font=\scriptsize] at (a-1-1.-180) {$\ho_\Omega$};
\end{tikzpicture}\]
gives a special homotopy retraction. We can further require that $\ho_\Omega|_{\Omega^\bullet(\fix(g))_G} = 0$ for $g \neq e$ in $2$-dimensional Calabi-Yau cases. Then (\ref{eq: undeformed homotopy}) gives a special homotopy retraction.

\begin{rem}There are many choices for $(i,p,\ho_\Omega)$ and none is canonical. For example, for $W = x^n+y^m$ with $G$ trivial on $\mathds{C}^2$, we can select such an $\ho_\Omega$ as
\begin{align}
\ho_\Omega(x^ay^b \df x\wedge \df y) \coloneqq & \begin{dcases}
\f{1}{n} x^{a-n+1}y^b\df y, & a\geqslant n-1,\\
-\f{1}{m} x^a y^{b-m+1} \df x,& a< n-1, b\geqslant m-1,\\
0, & \text{ else,}
\end{dcases}\\
\ho_\Omega ( x^ay^b \df x)\coloneqq &
\begin{dcases}
\f{1}{n} x^{a-n+1}y^b,&a\geqslant n-1,\\
0,& a<n-1,
\end{dcases}\\
\ho_\Omega ( x^ay^b \df y)\coloneqq & 0,
\end{align}
with the easiest inclusion $i$ and projection $p$.

\end{rem}
Finally, we can regard the deformed differential $\tilde{\p}_{b(\bm{\tau},\bm{s})}+u\tilde{B}$ as a perturbation of $\tilde{\p}_{b}$. Then if $G$ satisfies the Calabi-Yau condition that $G\subset SL(2,\mathds{C})$, we have a special homotopy retraction,
\begin{equation}\label{eq: deformed homotopy}
\begin{tikzpicture}
\matrix (a) [matrix of math nodes, row sep=3.5em, column sep=2em, text height=2ex, text depth=0.9ex]
{\big( C_{\bullet}^c(\bar{A},A[G])_G[[\bm{\tau},\bm{s}]]((u)),\tilde{\p}_{b(\bm{\tau},\bm{s})}+u\tilde{B}\big) & \big( \bigoplus\limits_{g\in G}\jac(W_g)_G[\overline{2-\abs{\bm{I}_g}}][[\bm{\tau},\bm{s}]]((u)),0\big).\\};
\path[-latex,semithick]
(a-1-1) edge (a-1-2);
\path[latex-,semithick]
(a-1-1.-2) edge (a-1-2.-178);
\draw[-latex,semithick]
(a-1-1.-176)  arc (-40:-320:11pt) -- ++ (.8pt,-.8pt);
\end{tikzpicture}
\end{equation}

The induced differential on the right complex is zero because it is concentrated in even degrees.
\begin{rem} We only construct the homotopies in two dimensional cases. However, they all can generalized in any dimensions.
\end{rem}
\end{appendices}

\Addresses

\end{document}